\documentclass[12pt]{article}
\usepackage{amsmath, verbatim}
\usepackage{amssymb}
\usepackage[margin=1in]{geometry}
\usepackage{tikz}
\usepackage{soul}
\usetikzlibrary{decorations.markings}
\usetikzlibrary{decorations.pathreplacing}
\allowdisplaybreaks
\usepackage{chngcntr}
\counterwithin{figure}{section}
\usepackage{authblk}
\usepackage{amsthm}
\usepackage{mathtools} 
\newcommand{\E}{\mathbf{E}}
\newcommand{\R}{\mathbb{R}}
\newcommand{\rr}{\mathbb{R}}
\newcommand{\N}{\mathbb{N}}

\newcommand{\x}{\mathbf{x}}

\newcommand{\calF}{\mathcal{F}}

\newcommand{\calm}{\mathcal{M}}
\DeclareMathOperator{\tr}{\mathrm{tr}}

\title{Identifiability of phylogenetic parameters from $k$-mer data under the coalescent}
\author{Chris Durden}
\author{Seth Sullivant}
\affil{Department of Mathematics, North Carolina State University}

\newtheorem{theorem}{Theorem}[section]
\newtheorem{lemma}[theorem]{Lemma}
\newtheorem{proposition}[theorem]{Proposition}
\newtheorem{corollary}[theorem]{Corollary}

\theoremstyle{definition}
\newtheorem{definition}[theorem]{Definition}

\newtheorem{example}[theorem]{Example}

\begin{document}
\tikzset{middlearrow/.style={
        decoration={markings,
            mark= at position 0.5 with {\arrow{#1}} ,
        },
        postaction={decorate}
    }
}
\maketitle
\begin{abstract}
    Distances between sequences based on their $k$-mer frequency counts can be used to reconstruct phylogenies without first computing a sequence alignment.
    Past work has shown that effective use
    of $k$-mer methods depends on 1)  model-based corrections to
distances based on $k$-mers and 2) breaking long sequences into blocks to
obtain repeated trials from the  sequence-generating process.  
Good performance of such methods is based on having many
high-quality blocks with many homologous sites, which can be problematic to 
guarantee a priori.  

Nature provides natural blocks of sequences into homologous
regions---namely, the genes.  However, directly using past work in this setting
is problematic because of possible discordance between different gene trees
and the underlying species tree.    Using the multispecies coalescent model as a basis,
we derive model-based moment formulas that involve the divergence times and the 
coalescent parameters.  From this setting, we prove  identifiability results for the tree
 and branch length parameters  under the Jukes-Cantor model of sequence mutations.
\end{abstract}

\section{Introduction}

Phylogenetic tree reconstruction methods that compare character states at homologous sites of molecular sequences require sequence alignment methods that identify the homologous sites.
    Many sequence alignment methods are progressive---that is, they first compute alignments for sequences from a subset of closely-related taxa and then assemble these results to generate an alignment for all of the taxa.
In this context, a guide tree provides the information about relatedness which is used to control the order in which alignments are assembled.
Phylogenetic reconstruction methods which do not rely on aligned sequences are needed to construct such guide trees.

Although multiple sequence alignment algorithms have been designed to reflect an insertion and deletion process that occurs along a phylogenetic tree, most current methods for constructing the guide tree itself are not based on any explicit models of evolution.
    We therefore aim to develop methods based on widely-used evolutionary modeling assumptions that can be used in the context of multiple sequence alignment to reconstruct a phylogenetic tree.
A natural approach is to record the k-mer frequencies of sequences and to assess relatedness among taxa by computing distances based on these frequencies.

For example, the guide trees computed in 
MUSCLE \cite{Edgar2004b, Edgar2004a} and Clustal Omega \cite{Blackshields2010, Sievers2011}
are based on using k-mer frequencies.  However, those k-mer methods are ad hoc, in the sense
that they are not derived based on any evolutionary modeling assumptions.
We take the point of
view that a desirable property of any phylogenetic method is that it should be
\emph{statistically consistent} under widely used phylogenetic modeling assumptions.  
That is,
if the method is applied to data from a standard model, the method should reconstruct
the correct tree with probability tending to $1$ as the amount of data increases.

In past work, Allman, Rhodes and the second author \cite{allman2017statistically}
devised a statistically consistent $k$-mer method for phylogenetic tree reconstruction
based on models for point substitutions only, and without an insertion and deletion (indel) process.  
This method utilizes a model-based correction to the $k$-mer distances between pairs of sequences.
A key idea in \cite{allman2017statistically}, originating in the work of Daskilaskis and Roch  
\cite{Daskalakis2013}, is to break sequences into blocks to get a distribution of $k$-mer
distances.  From this distribution, various features of the underlying substitution
process can be extracted.  More specifically, to compute a $k$-mer distance between
two sequences $S_1$ and $S_2$, we subdivide each sequence into $r$ subsequences
$S_{11}, \ldots, S_{1r}$, and $S_{21}, \ldots, S_{2r}$, respectively.  The subdivision
is chosen so that, for each $i$, $S_{1i}$ and $S_{2i}$ consist of mostly orthologous sites.
Then, to each pair of subsequences $S_{1i}$ and $S_{2i}$, we compute $k$-mer vectors
$X_{1i}$ and $X_{2i}$, compute the appropriate $k$-mer distance between $X_{1i}$ and $X_{2i}$
and average the results from $i = 1, \ldots, r$.

This procedure greatly increases the 
accuracy of the $k$-mer method if we assume all sequences are drawn from the same underlying
phylogenetic process, because the average of independent draws from the same underlying
process converges to the true underlying parameter being estimated, by the law of large numbers.
Since we do not know a priori exactly when parts of the sequence are orthologous, a
subdivision into a small number of parts $r$ of roughly equal lengths guarantees that most
sites in the two sequences $S_{1i}$ and $S_{2i}$ are orthologous.  
Heuristic
modification of this procedure is used in the case when the sequences are generated from 
a model with an indel process.  
Even with a moderate
indel process at work, the number of blocks used cannot be too big without running into trouble, since the number of orthologous sites within each block becomes insufficient.

Fortunately, nature provides blocks that automatically correspond one to another---namely, the genes.
However, comparing sequences from many different genes requires the analysis
of more complex probabilistic models, describing how the evolutionary histories of genes are related to the history of the species from which they are derived.
The phylogenetic history of a set of species and the history of any given set of orthologous genes from those species are represented by a species tree and a gene tree, respectively.
In particular, it is well-known that
gene trees need not be the same, and that gene trees need not be equal to the
underlying species tree \cite{pamilo1988relationships}.  The simplest model to handle this 
discrepancy is the \emph{multispecies coalescent model}.
Given a species tree with branch lengths, the multispecies coalescent gives a probability
distribution on gene trees \cite{takahata1995divergence}.  
With such a gene tree, we can then use a standard model of 
sequence evolution to describe the probability distribution of gene sequences.

Our goal in this paper is to extend the results of \cite{allman2017statistically}
to the more general setting from the previous paragraph where blocks correspond to
gene sequences, generated by a mutation process on gene trees that come from the multispecies coalescent.
In Section \ref{sec:expected}, we derive a generalization of the main formula from \cite{allman2017statistically},
which is a calculation of the expected squared Euclidean distance between $k$-mer vectors
of sequences, when the underlying gene trees are generated randomly from the coalescent process.  
Using this formula, in Sections \ref{sec:local} and \ref{sec:treetopo} we prove identifiability results on the underlying model
parameters (unrooted species tree and numerical parameters) which is the first step
towards developing a statistically consistent method based on $k$-mers.  
Section \ref{sec:kand1} contains some further identifiability results where combinations of $k$-mer
vectors of different sizes are used.  We conclude in Section \ref{sec:conclusion} with
a discussion of further directions and ideas about how our identifiability results might
lead to new algorithms for constructing trees from $k$-mer data.


\section{Expected $k$-mer Distances}\label{sec:expected}

In this section, we give a derivation of the expected Euclidean distance between
$k$-mer vectors of sequences when their gene tree is generated by the coalescent model and 
mutations arise via any stationary Markov model.  We also present the special form of
this distance in the case of the Jukes-Cantor substitution model, from which we derive
our later results.

We first define what we mean by the $k$-mer vector of a sequence.
As mentioned above, the $k$-mers of a given sequence are the subsequences of length $k$.
We form a $k$-mer vector by recording the number of times each $k$-mer occurs in the given sequence.
More precisely, let $S = s_1s_2 \cdots s_m \in [L]^m$ be a sequence of length $m$ in the alphabet $[L] = \{1, \ldots, L\}$.
For $k \leq m$, a $k$-mer is a subsequence $s_ps_{p+1} \cdots s_{p+k-1}$ for some starting position $p \in \{1, \ldots, m-k+1\}$.
The $k$-mer count vector $X$, associated to the sequence $S$, is the
vector of length $L^k$ whose coordinates are indexed by the words $W \in [L]^k$, and where the component $X^W$, corresponding to word $W$, records the number of times $W$ appears as a subsequence in $S$. 

We next consider two sequences $S_1$ and $S_2$ descended from a common ancestor, and we define a $k$-mer distance between them as follows.
First, we assume that each site in each sequence is generated independently by a Markov mutation process and that the distribution at each site is stationary.
Let $Q$ be the rate matrix of the mutation process, and let $\pi$ be the associated stationary distribution.
By stationarity, the probability of a $k$-mer in either sequence is $(\pi^W)_{W \in [L]^k}$, where $\pi^W = \prod_{i \in [k]} \pi^{w_i}$.
The $k$-mer distance between sequences $S_1$ and $S_2$ is then $\sum_{W \in [L]^k} \frac{1}{\pi^W}(X^W_1 - X^W_2)^2$.
Allman, Rhodes, and the second author previously derived the expected value of this $k$-mer distance for a pair of orthologous sequences with divergence time $\tau$ under such a mutation process.

\begin{proposition} \cite{allman2017statistically} \label{prop:allmanetc}
Let $S_1$ and $S_2$ be two sequences of length $m$ generated from an indel-free
Markov model with transition matrix $M = \exp(Q \tau)$, where $Q$ is
the rate matrix, and stationary initial distribution $\pi$.
Let $X_1$ and $X_2$ be the resulting $k$-mer count vectors.  Then
\begin{equation*}
    \E \left[ \sum_{W \in [L]^k} \frac{1}{\pi^W} (X^W_{1}-X^W_{2})^2  \right]
    = 2(m-k+1) (L^k-(\tr M)^k).
\end{equation*}
\end{proposition}
If $\lambda_1, \ldots, \lambda_L$ are the (not necessarily distinct) eigenvalues of $Q$, then the trace of the transition matrix $M = \exp(Q \tau)$ can be computed from these eigenvalues: 
\begin{align*}
    \tr M = \sum_{l=1}^L e^{\lambda_l \tau}.
\end{align*}

Note that, to use Proposition \ref{prop:allmanetc} in practice, we estimate the expected $k$-mer distance by dividing the sequence into blocks and computing an average over the blocks, as discussed in the introduction.
However, to discuss the value of the expectation, we only work with a single
sequence (or, equivalently, a single block) here and throughout. Thus, we have suppressed the index $i$ that we used in the introduction to refer to individual blocks.


We next generalize the expected $k$-mer distance formula of Proposition \ref{prop:allmanetc} 
 by allowing the divergence times of the sequences to vary according to the multispecies coalescent.
 We start with a brief overview of the coalescent and multispecies coalescent models.
 For our purposes it is sufficient to consider a special case of the multispecies coalescent model, consisting of only two species.
 We thus describe the two-species coalescent model in detail and we use this to derive our coalescent-based $k$-mer distance formula.

 The original \textit{$n$-coalescent}, as formulated by Kingman in 1982 \cite{kingman1982coalescent}, is a stochastic (Markov) process that represents the hierarchical history of family relationships of a set of $n$ objects.
 These objects are said to coalesce when two blocks of the partition merge to form a single larger block.
 In the context of this paper, the coalescent process provides a stochastic model of the genealogical history of $n$ molecular sequences.
 A realization of the coalescent process gives a gene tree, which represents the pattern of the coalescence events among the lineages as they extend back in time to their most recent common ancestor.
The $n$-coalescent is commonly used to model gene trees of sequences, based on the assumption that they are sampled from a single, large, randomly-mating population.
In this model, all possible coalescence events (associated with all pairs of distinct lineages) occur at the same rate, which depends on the population size.

The \textit{multispecies coalescent model} with $n$ species extends Kingman's $n$-coalescent model.
In the multispecies coalescent process, coalescence events are constrained according to a species tree parameter, so that the process only generates gene trees that are consistent with the phylogenetic relationships given by the species tree.
Specifically, genes derived from two taxa cannot coalesce more recently than the taxa themselves diverged, as delineated by the species tree.
To represent this pictorially, the branches of the species tree are drawn with thick branches, and gene trees are drawn embedded within the species tree, as in Figure \ref{fig:twospeciescoalescent}.
We will assume, for simplicity, that there is a one-to-one correspondence between the leaves of the species tree and the leaves of any embedded gene tree.

Once a population size is assigned to each ancestral taxon in the species tree, the multispecies coalescent model determines a probability distribution on the gene trees.
The full probability distribution of gene trees under the $n$-species coalescent model has been described by Rannala and Yang \cite{rannala2003bayes}, but the $2$-species case is sufficient for our purposes, since we only calculate expected values of $k$-mer distances for one pair of species at a time.
The basic structure of the two-species coalescent model, showing the gene tree embedded within a species tree with two taxa, is illustrated in Figure \ref{fig:twospeciescoalescent}.

The expected $k$-mer distance between two orthologous sequences from two species only depends on the divergence time for the corresponding two-leaf gene trees.
We thus use the distribution of the divergence time, based on the two-species coalescent model, along with Proposition \ref{prop:allmanetc} to derive our generalized expected $k$-mer distance.
After computing the expected $k$-mer distance for each species pair, we will assemble these pairwise distances together to obtain a collection of expected $k$-mer distances, which we would like to use to estimate phylogenetic parameters.

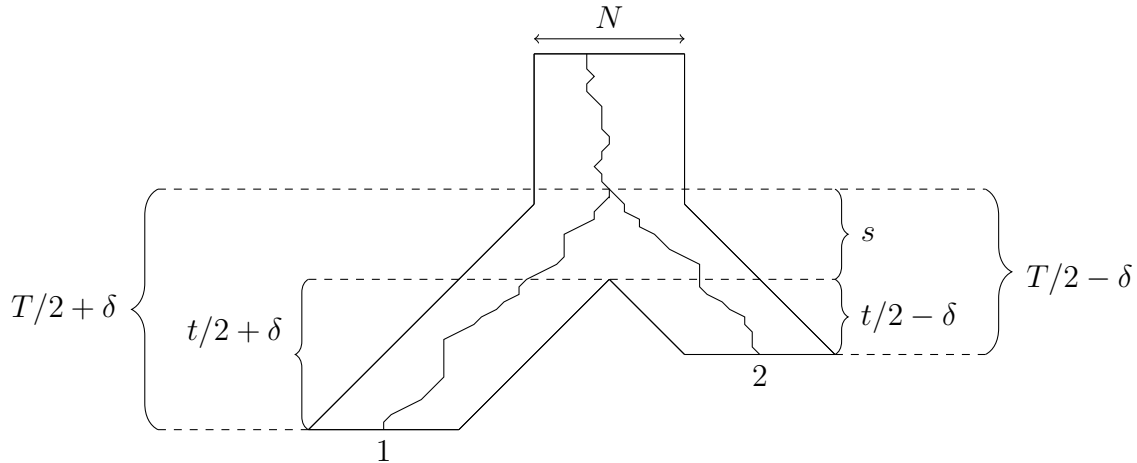
\begin{figure}
\begin{centering}
\begin{tikzpicture}[scale=4]
        \draw[fill] (0.5,0) -- (1,0);
        \draw[fill] (1,0) -- (0.5,0.5);
        \draw[fill] (0.5,0.5) -- (0.5,1);
        \draw[fill] (0.5,1) -- (0,1);
        \draw[fill] (0,1) -- (0,0.5);
        \draw[fill] (0,0.5) -- (-0.75,-0.25);
        \draw[fill] (-0.75,-0.25) -- (-0.25,-0.25);
        \draw[fill] (-0.25,-0.25) -- (0.25,0.25);
        \draw[fill] (0.25,0.25) -- (0.5,0);

        \draw [decorate,decoration={brace,mirror,amplitude=5pt},xshift=0pt,yshift=0pt]
        (1.00,0) -- (1.00,0.25);
        \node[anchor=west] at (1.05,0.125) {$t/2-\delta$};
        \draw [decorate,decoration={brace,mirror,amplitude=5pt},xshift=0pt,yshift=0pt]
        (1.00,0.25) -- (1.00,0.55);
        \node[anchor=west] at (1.05,0.4) {$s$};

        \draw [decorate,decoration={brace,mirror,amplitude=10pt},xshift=0pt,yshift=0pt]
        (1.5,0) -- (1.5,0.55);
        \node[anchor=west] at (1.6,0.25) {$T/2-\delta$};

        \draw [decorate,decoration={brace,amplitude=5pt},xshift=0pt,yshift=0pt]
        (-0.75,-0.25) -- (-0.75,0.25);
        \node[anchor=east] at (-0.8,0.075) {$t/2+\delta$};

        \draw [decorate,decoration={brace,amplitude=10pt},xshift=0pt,yshift=0pt]
        (-1.25,-0.25) -- (-1.25,0.55);
        \node[anchor=east] at (-1.35,0.15) {$T/2+\delta$};

        \draw[dashed] (-0.75,0.25) -- (1,0.25);
        \draw[dashed] (-1.25,0.55) -- (1.5,0.55);
        \draw[dashed] (-1.25,-0.25) -- (-0.75,-0.25);
        \draw[dashed] (1,0) -- (1.5,0);

        \draw[<->]
        (0,1.05) -- (0.5,1.05);
        \node[anchor=south] at (0.25,1.05) {$N$};
        \node[anchor=north] at (-0.5,-0.25) {$1$};
        \node[anchor=north] at (0.75,0) {$2$};

\pgfmathsetseed{7}
\draw (-0.5,-0.25) 
        \foreach \i 
        [ 
                evaluate=\angle using round(2*rnd-0.5)*180, 
                evaluate=\x using 0.025*(1+round(2*(rnd-0.5))), 
                evaluate=\y using 0.025 
        ] 
        in {1,...,31}{ -- ++(\x, \y) }; 
\draw (0.75,0) 
        \foreach \i 
        [ 
                evaluate=\angle using round(2*rnd-0.5)*180, 
                evaluate=\x using 0.025*(-1+round(2*(rnd-0.5))), 
                evaluate=\y using 0.025 
        ] 
        in {1,...,22}{ -- ++(\x, \y) }; 
\draw (0.25,0.52) -- (0.25,0.55);
\draw (0.25,0.55) 
        \foreach \i 
        [ 
                evaluate=\angle using round(2*rnd-0.5)*180, 
                evaluate=\x using 0.025*(round(2*(rnd-0.5))), 
                evaluate=\y using 0.025 
        ] 
        in {1,...,18}{ -- ++(\x, \y) }; 

\end{tikzpicture}
\caption{The two-species coalescent model. The quantity $\delta$ is introduced solely for the purpose of illustrating that the species tree is not ultrametric. The species divergence time $t$ represents the total evolutionary time separating species $1$ and $2$. The sequence divergence time $T$ represents the total evolutionary time separating a pair of orthologous sequences from species $i$ and $j$. The time between speciation and coalescence is denoted by $s$. In the coalescent model the distribution of $s$ depends on the ancestral population size $N$.}
    \label{fig:twospeciescoalescent}
\end{centering}
\end{figure}

We now formulate precisely the model that we use to derive the expected $k$-mer distance between sequences when the gene tree varies according to the multispecies coalescent.
We consider two species, and we let $t$ denote the speciation time, in true time units.
This time represents the total evolutionary time separating the species.
Thus it is given by the sum of the branch lengths along the species tree leading from the point where the species diverge to each of the two leaves (See Figure \ref{fig:twospeciescoalescent}).
The times along each of these branches are not assumed to be equal.
In other words, we do not assume that our species tree is ultrametric.
We consider pairs of orthologous sequences, where each pair consists of a sequence from one species along with an ortholog from the other species.
The divergence time of any pair of orthologs will exceed the speciation time $t$ by some amount $2s$, representing twice the time between speciation and coalescence \cite{takahata1986attempt}.
Thus, the sequence divergence time is given by $\tau = t+2s$.
The coalescence time $s$ depends on the size of the population ancestral to the two species, which we denote by $N$.
A graphical representation of the quantities in this model is shown in Figure \ref{fig:twospeciescoalescent}.

To generalize Proposition \ref{prop:allmanetc}, we allow the gene trees of orthologous sequences to vary according to the multispecies coalescent, so that their divergence times $\tau$ are realizations of a random variable $T$.
Under this model the divergence time is given by $T = t+2s$, where the coalescence time $s$ is exponentially distributed.
Since the size of the ancestral population is $N$, the rate of coalescence is $1/N$, and 
\begin{align*}
    s \sim \text{Exp}(1/N).
\end{align*}

We compute the expected $k$-mer distance with respect to the distribution of 
$T$, by interpreting Proposition \ref{prop:allmanetc} as the expected $k$-mer 
distance given that $T = \tau$, and using the law of total expectation:
\begin{align*}
    &\E \left[ \sum_{W \in [L]^k} \frac{1}{\pi^W} (X^W_{1}-X^W_{2})^2 \right] \\
    &         \qquad =\E \left[\E \left[ \sum_{W \in [L]^k} \frac{1}{\pi^W} (X^W_{1}-X^W_{2})^2 ~\bigg|~ T = t + 2s \right]\right]\\
    &         \qquad = \E \left[2(m-k+1) \left(L^k-\left(\sum_{l=1}^L\exp(\lambda_l (t+2s))\right)^k\right)\right]\\
    &         \qquad = 2(m-k+1) \left(L^k-\sum_{q_1+\ldots+q_L=k} { k \choose q_1, \ldots, q_L } \E \left[\exp\left((\lambda_1 q_1 + \ldots + \lambda_L q_L) (t+2s)\right)\right]\right)\\
    &         \qquad = 2(m-k+1) \bigg(L^k-\sum_{q_1+\ldots+q_L=k} { k \choose q_1, \ldots, q_L } \exp\left((\lambda_1 q_1 + \ldots + \lambda_L q_L)t\right)\\
    &\qquad \qquad \qquad \qquad \qquad \qquad \qquad \qquad \qquad \qquad \qquad \cdot \E \left[\exp(2 (\lambda_1 q_1 + \ldots + \lambda_L q_L) s)\right]\bigg)\\
    &\qquad  = 2(m-k+1) \bigg(L^k-\sum_{q_1+\ldots+q_L=k} { k \choose q_1, \ldots, q_L } \exp((\lambda_1 q_1 + \ldots + \lambda_L q_L)t)\\
    &\qquad \qquad \qquad \qquad \qquad \qquad \cdot \int_0^\infty\exp(2 (\lambda_1 q_1 + \ldots + \lambda_L q_L) s) \frac{1}{N} \exp(-s/N) ds\bigg)\\
    &\qquad = 2(m-k+1) \left(L^k-\sum_{q_1+\ldots+q_L=k} { k \choose q_1, \ldots, q_L } \frac{\exp((\lambda_1 q_1 + \ldots + \lambda_L q_L)t)}{1-2N(\lambda_1 q_1 + \ldots + \lambda_L q_L)}\right).
\end{align*}

The preceding calculation proves the following proposition.

\begin{proposition}\label{prop:general}
Let $S_1$ and $S_2$ be sequences of length $m$ 
with divergence time distributed according to the multispecies coalescent process with population size $N$ and species divergence time $t$. 
Suppose $S_1$ and $S_2$ are generated from an indel-free Markov model with transition rate matrix $Q$ and stationary distribution $\pi$, and let $\lambda_1, \ldots, \lambda_L$ be the eigenvalues of $Q$.
Let $X_1$ and $X_2$ be the resulting $k$-mer count vectors.
Then
\begin{multline*}
\E\left[ \sum_{W \in [L]^k} \frac{1}{\pi^W} (X^W_{1}-X^W_{2})^2 \right]  =  \\
2(m-k+1) \left(L^k-\sum_{q_1+\ldots+q_L=k} { k \choose q_1, \ldots, q_L } \frac{\exp((\lambda_1 q_1 + \ldots + \lambda_L q_L)t)}{1-2N(\lambda_1 q_1 + \ldots + \lambda_L q_L)}\right).
 \end{multline*}
\end{proposition}

In practice, in order to obtain precise estimates of the expected $k$-mer distance of Proposition \ref{prop:general}, we would like to compute an average $k$-mer distance over many independently-generated sequences.
We have in mind the context in which a sequence (e.g. a genome) is divided into blocks corresponding to genes, and the expected $k$-mer distance is estimated by taking the average $k$-mer distance over all blocks.
Under our model, the sequence data from distinct loci are independent if the gene trees for those loci are independent, given the species tree.
We expect this to be the case for unlinked loci, but it may also hold for some linked loci if the population sizes, population structure, and recombination rates are sufficient to decouple inheritance \cite{wakeley2009coalescent, mcvean2002genealogical}.

In order to use Proposition \ref{prop:general} to derive statistically-consistent model-based corrections to the $k$-mer distances, the map from parameter space to the collection of all $k$-mer distances must be one-to-one.
For the remainder of this paper, we study this problem for the special case where $L=4$ and $Q$ is the Jukes-Cantor rate matrix
\begin{align*}
    Q = \left(\begin{array}{rrrr} -\mu & \mu/3 & \mu/3 & \mu/3\\
 \mu/3 &-\mu & \mu/3 & \mu/3 \\
\mu/3 & \mu/3 & -\mu & \mu/3 \\
\mu/3 & \mu/3 & \mu/3 & -\mu 
    \end{array}\right).
\end{align*}
The eigenvalues of $Q$ are $\lambda_1 = 0$ and $\lambda_2 = -4\mu/3$, with multiplicities $1$ and $3$, respectively.
The stationary distribution of the corresponding Markov process is $\pi^W = \frac{1}{4^k}$ for all $W \in [4]^k$.
So Proposition \ref{prop:general} reduces to the following:
\begin{corollary}
Let $S_1$ and $S_2$ be sequences of length $m$ 
with divergence time distributed according to the multispecies coalescent process with population size $N$ and species divergence time $t$. 
Suppose $S_1$ and $S_2$ are generated from an indel-free Jukes-Cantor mutation model.
Let $X_1$ and $X_2$ be the resulting $k$-mer count vectors.
Then
\[
    \E[ \| X_1 - X_2\|_2^2 ] \, \,  =  \, \, 
    2(m-k+1) \left(1- \frac{1}{4^k}\sum_{h=0}^k { k \choose h } 3^h \frac{\exp(-4h \mu  t/3)}{1+8h \mu N /3}\right).
\]
\label{cor:jukescantor}
\end{corollary}

To analyze the dependence of the expected $k$-mer distance on the parameters $t,N,$ and $\mu$, we define a function which maps these parameters to the expected $k$-mer distance:

\begin{align*}
    f_k(t,N,\mu)&= 2(m-k+1)\left(1- \frac{1}{4^k}\sum_{h=0}^k { k \choose h } 3^h\frac{\exp(-4h\mu t/3)}{1+8h \mu N/3}\right).
\end{align*}

We note that
\begin{align}
    f_k(\alpha t, \alpha N, \mu/\alpha) = f_k(t,N,\mu) \label{eqn:invariance}
\end{align}
for any $\alpha > 0$.
This can be seen as a specific consequence of the fact that both the mutation process and the coalescence process are invariant under similar transformations of the parameters $t, \mu$ and $N$: These processes are unaffected by changing evolutionary times from $t$ to $\alpha t$, while simultaneously changing the coalescence rate from $1/N$ to $1 / \theta = 1/ N \alpha$ and the mutation from from $\mu$ to $\mu / \alpha$.
Thus, we have a fundamental inability to determine the values of the parameters $\mu$, $N$, and $t$ from data generated by these processes. 

We use a standard approach of adjusting time units to suppress the invariance described by Equation \ref{eqn:invariance}: We assume that time is measured in substitution units---that is, time units in which $\mu = 1$. Then $t$ represents time in units of the expected number of substitutions. In these time units, the rate of coalescence is $1/N \mu$. We let $\theta$ denote the quantity $N\mu$ in the denominator. With respect to the new time units, we obtain the expected $k$-mer distance as a function of $t$ and $\theta$:

\begin{align*}
    f_k(t,\theta)    &= 2(m-k+1)\left(1- \frac{1}{4^k}\sum_{h=0}^k { k \choose h } 3^h\frac{\exp(-4h t/3)}{1+8h\theta/3}\right).
\end{align*}

To facilitate an algebraic analysis of model identifiability, we reparametrize $f_k$ by introducing a quantity $x$, which we call the transformed species divergence time, defined by the following equation:
\begin{align*}
    x = \text{exp}(-4 t /3),
\end{align*}
where the divergence time $t$ is given in substitution units, as mentioned previously.
With this change of variables the expected $k$-mer distance of Corollary \ref{cor:jukescantor} can be written as a rational function of $x$ and $\theta$:
\begin{align}
    f_k(x,\theta)    &= 2(m-k+1)\left(1- \frac{1}{4^k}\sum_{h=0}^k { k \choose h } 3^h \frac{x^h}{1+8h \theta /3}\right). \label{eqn:rationalparam}
\end{align}

The parametrization of the expected $k$-mer distance in Equation \ref{eqn:rationalparam} depends continuously on two parameters, so these parameters are not identifiable from a single expected $k$-mer distance between one pair of species. We thus consider a situation in which we have $n$ species ($n>2$). For any two distinct species, labeled by $i,j \in [n]$, we consider $k$-mer vectors $X_i$ and $X_j$ generated by the coalescent and mutation processes. 
We let $t_{ij}$ denote the species divergence time, and we let $x_{ij}$ be the corresponding transformed divergence time. 
We suppose that the ${n \choose 2}$ $k$-mer distances between pairs of species are parametrized by a common set of parameters which, taken together, describe the evolutionary history of these species. 
The simplest assumption which provides a model that is potentially identifiable---having fewer than ${n \choose 2}$ independent parameters---is that the species divergence times $t_{ij}$ are distances among the leaves of a species tree, and all of the ancestral populations have the same size $N$ (and therefore the same value of $\theta$). 
Under these assumptions, we can parametrize all ${n \choose 2}$ pairs of expected $k$-mer distances by $2n-2$ parameters: $2n-3$ parameters giving the branch lengths of the species tree, and one $\theta$ parameter.

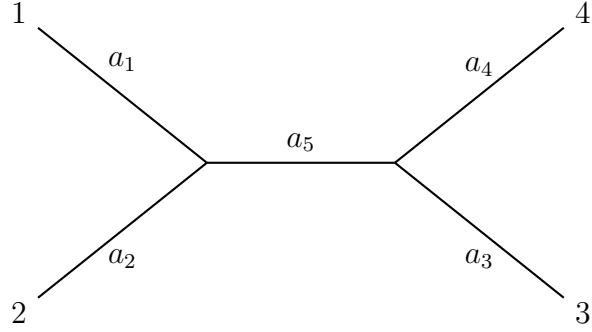
\begin{figure}
\begin{center}
\begin{tikzpicture}[
        thick,
        level/.style={level distance=2.5cm},
        level 2/.style={sibling distance=4cm},
        level 3/.style={sibling distance=2.5cm}
    ]
    \coordinate
        child[grow=left]{
            child {
                node {$1$}
                edge from parent node [above=5pt] {$a_1$}
            }
            child {
                node {$2$}
                edge from parent node [below=3pt] {$a_2$}
            }
            edge from parent node [above] {$a_5$}
        }
        child[grow=right, level distance=0pt] {
            child {
                node {$3$}
                edge from parent node [below=3pt] {$a_3$}
            }
            child {
                node {$4$}
                edge from parent node [above=3pt] {$a_4$}
            }
            edge from parent 
    };
\end{tikzpicture}
\end{center}
\caption{Example labeling of a $4$-taxon tree with transformed edge weight parameters $(a_e)_{e \in P(i,j)}$. The parameters are defined by $a_e = \exp(-4 w_e/3)$ where $w_e$ represents the evolutionary time along edge $e$ in substitution units. In the parametrization of our model, the transformed species divergence time for taxa $1$ and $3$, for example, is $x_{13} = a_1a_5a_3$.}
\label{parameters}
\end{figure}

We now describe this parametrization. Suppose $T$ is a species tree with leaves labeled by the taxa $[n]$, and let ${\{w_e \in \R_{\ge 0} : e \in E(T)\}}$ be an edge weighting of $T$, giving evolutionary times along edges of $T$.
If the divergence times $(t_{ij})_{1\le i <j \le n}$ are the distances between leaves $i$ and $j$ of $T$, then they are obtained by summing edge weights along paths in $T$:
\begin{align}
t_{ij} = \sum_{e \in P(i,j)} w_e, \label{eqn:branchlengths}
\end{align}
 where $P(i,j)$ is the set of edges of the unique path in $T$ connecting leaves $i$ and $j$.
 The transformed divergence times $(x_{ij})_{1\le i < j \le n}$ can then be expressed analogously: For each edge $e \in E(T)$, we define a transformed edge weight 
$a_e = \exp(-4 w_e/3).$
Then the transformed divergence times are obtained by multiplying branch lengths along corresponding paths in $T$:
\begin{align*}
x_{ij} = \prod_{e \in P(i,j)} a_e. \label{eqn:tranformedbranchlengths}
\end{align*}
 Since $w_e \ge 0$, we have $a_e \in (0,1]$. In the analyses which follow, we will parametrize the expected $k$-mer distance between species $i$ and $j$ using these transformed edge weights: $\E[ \| X_i - X_j\|_2^2 ] = f_k(\prod_{e \in P(i,j)} a_e,\theta)$.

We note that Proposition \ref{prop:general} gives a $k$-mer distance formula that depends on the total evolutionary distance between the species, and not on the individual branch lengths leading to their common ancestor.
Thus the $k$-mer distances of Proposition \ref{prop:general} do not depend on the choice of a root for the tree $T$.
The following observations clarify why this is true.
First, the distribution of the distances between the leaves of the gene tree does not depend on where the species tree is rooted.
This can be seen by observing that changing the root is equivalent to varying the quantity $\delta$ in Figure \ref{fig:twospeciescoalescent}, which does not affect any of the divergence times.
Second, the expected $k$-mer distances, for a fixed gene tree, as given by Proposition \ref{prop:allmanetc} also do not depend on the choice of a root for the gene tree.
In fact, by the stationarity of the Markov mutation process, the joint distribution of the sequences also do not depend on the root.
Together, these observations imply that the rooted tree is not identifiable from $k$-mer distances under our model. 
Thus we will consider only the unrooted tree when we formulate our identifiability results.

For a given $n$-leaf binary tree $T$ and for each $k \in \N$, we construct a rational
map whose inputs are the parameter $\theta$ and the transformed branch lengths
$a_e = \exp(-4\mu w_e/3)$, and which maps to the ${n \choose 2}$ expected $k$-mer distances
between pairs of species.  This map, denoted by $\phi_{k,T}$,
has the following form:
\[
\phi_{k,T}  :  \rr^{2n -2}  \rightarrow   \rr^{n(n-1)/2}
\]
\[
(a_e : e \in E(T),  \theta)  \mapsto    \left(f_k\left( \prod_{e \in P(i,j)} a_e,  \theta\right)\right)_{1 \leq i < j \leq n}.
\]

The parameters $a_e = \exp(-4 w_e/3)$ are 
restricted to lie in $(0,1]$, and the parameter $\theta = \mu N$ is positive. 
The parameter space of $\phi_{k,T}$ is thus $\Theta_T = (0,1]^{2n-3} \times \R_{> 0}$.
We note that the smallest $n$ for which the dimension of the image is at least the number of parameters is $n=4$. For $n=4$, $\phi_{k,T}: (0,1]^5 \times \R_{> 0} \to \R^6$, so that it is in principle possible for the parameters to be identifiable. We will see in the next section that the parameters are in fact locally identifiable in this case if $k > 1$.

\begin{example}
Figure \ref{parameters} shows the set of edge length parameters 
$(a_e)_{e \in P(i,j)}$ arising from a particular labeling of the $5$ edges of a $4$-taxon tree.  
Suppose that $k = 2$.  Then
$\phi_{k,T}$ takes the following form:
\[
\begin{pmatrix}
a_1 \\
a_2 \\
a_3 \\
a_4 \\
a_5 \\
\theta
\end{pmatrix}
\mapsto
\begin{pmatrix*}[l]
f_{2,T}(a_1a_2, \theta)  \\
f_{2,T}(a_1a_5a_3, \theta) \\
f_{2,T}(a_1a_5a_4, \theta) \\
f_{2,T}(a_2a_5a_3, \theta) \\
f_{2,T}(a_2a_5a_4, \theta) \\
f_{2,T}(a_3a_4, \theta)
\end{pmatrix*} 
=
\begin{pmatrix*}[l]
2(m-1)( 1 - \frac{1}{16} ( 1 + \frac{18a_1 a_2}{3 + 8 \theta} + \frac{27 a_1^2 a_2^2}{3 + 16 \theta}) \\
2(m-1)( 1 - \frac{1}{16} ( 1 + \frac{18a_1 a_5 a_3}{3 + 8 \theta} + \frac{27 a_1^2 a_5^2 a_3^2}{3 + 16 \theta}) \\
2(m-1)( 1 - \frac{1}{16} ( 1 + \frac{18a_1 a_5 a_4}{3 + 8 \theta} + \frac{27 a_1^2 a_5^2 a_4^2}{3 + 16 \theta}) \\
2(m-1)( 1 - \frac{1}{16} ( 1 + \frac{18a_2 a_5 a_3}{3 + 8 \theta} + \frac{27 a_2^2 a_5^2 a_3^2}{3 + 16 \theta}) \\
2(m-1)( 1 - \frac{1}{16} ( 1 + \frac{18a_2 a_5 a_4}{3 + 8 \theta} + \frac{27 a_2^2 a_5^2 a_4^2}{3 + 16 \theta}) \\
2(m-1)( 1 - \frac{1}{16} ( 1 + \frac{18a_3 a_4}{3 + 8 \theta} + \frac{27 a_3^2 a_4^2}{3 + 16 \theta}) \\
\end{pmatrix*}
\]
\end{example}


\section{Local Identifiability from k-mer Distances}\label{sec:local}

In this section, we prove our first main identifiability result. 
We  show that the map $\phi_{k,T}$ is 
generically finite-to-one if $n \ge 4$ and $k>1$.
In other words, both $\theta$ and the branch length parameters of the phylogenetic
tree are locally identifiable from  $k$-mer distances provided $k > 1$ and there are at
least $4$ species.  This is proven in Theorem \ref{thm:localidentn4} and
Corollary \ref{cor:localident}.
We will use this result in the next section to show that the (unrooted) tree parameter is identifiable.

\begin{definition}
    A map $\phi: S \to \R^d$ defined on an open set $S \subset \R^n$ is called \textit{generically finite-to-one} if there exists a proper algebraic subset $\tilde{S} \subset \R^n$ such that the fiber 
\begin{align*}
    \calF_{\phi}(s') = \{ s \in S \mid \phi(s) = \phi(s') \}
\end{align*}
    is finite for all $s' \in S \backslash \tilde{S}$.
\end{definition}

The following folklore result is the standard tool to
prove local identifiability.  A recent proof can be found in \cite{leung2016identifiability}:

\begin{lemma}
 If $\phi: S \to \R^d$ is a polynomial or rational map defined on an open set $S \subset \R^n$ then $\phi$ is generically finite-to-one on $S$ if and only if the Jacobian matrix of $\phi$ generically has full column rank.
\label{rank_condition}
\end{lemma}

\begin{theorem}\label{thm:localidentn4}
For $k>1$ and $n=4$, the map $\phi_{k,T}: \R^6 \to \R^6$ is generically finite-to-one. In
particular, it is generically finite-to-one on $\Theta_{T} = (0,1]^5 \times \R_{> 0}$.
\label{local_identifiability}
\end{theorem}
    
\begin{proof}
        We show that the Jacobian matrix of $\phi_{k,T}$ has full column rank. We first write $\phi_{k,T}$ as a composite map and use the chain rule to write the Jacobian matrix as a product of matrices:
        $\phi_{k,T}$ can be expressed as the composition $\phi_{k,T} = F_{k,T} \circ \xi_T$ where
\begin{align*}
    F_{k,T} &= 
    (f_k\left( x_{12}, \theta \right), f_k\left( x_{13}, \theta \right), f_k\left( x_{14}, \theta \right), f_k\left( x_{23}, \theta \right), f_k\left( x_{24}, \theta \right), f_k\left( x_{34}, \theta \right)),\\
     \xi_T(a,\theta) &=
    \left(a_{1} a_{2}, a_{1} a_{3} a_{5}, a_{1} a_{4} a_{5}, a_{2} a_{3} a_{5}, a_{2} a_{4} a_{5}, a_{3} a_{4},\theta\right).
\end{align*}
(Here we have identified the input vector $(x_{12}, x_{13}, x_{14}, x_{23}, x_{24}, x_{34}, \theta)$ of $F_{k,T}$ with the output vector of $\xi_{T}$.)

The Jacobian matrix of the outer function $F_{k,T}$ is
\begin{equation}
    d_{(\x,\theta)} F_{k,T} =
    \left(\begin{array}{rrrrrrr}
            D_{12} & 0 & 0 & 0 & 0 & 0 & B_{12} \\
            0 & D_{13} & 0 & 0 & 0 & 0 & B_{13} \\
            0 & 0 & D_{14} & 0 & 0 & 0 & B_{14} \\
            0 & 0 & 0 & D_{23} & 0 & 0 & B_{23} \\
            0 & 0 & 0 & 0 & D_{24} & 0 & B_{24} \\
            0 & 0 & 0 & 0 & 0 & D_{34} & B_{34}
    \end{array}\right)
    \label{eqn:dF}
\end{equation}
where
    \begin{align*}
        D_{ij}(\x,\theta) &= \frac{\partial f_k}{\partial x_{ij}} (x_{ij},\theta) = -2(m-k+1)\frac{1}{4^k}\sum\limits_{h=1}^k {k \choose h} \frac{3^h h x_{ij}^{h-1}}{1+8\theta h/3}, \text{ and }\\
        B_{ij}(\x,\theta) &= \frac{\partial f_k}{\partial \theta} (x_{ij},\theta) = 2(m-k+1)\frac{1}{4^k}\sum\limits_{h=1}^k {k \choose h} \frac{3^h x_{ij}^{h} 8 h/3}{(1+8\theta h/3)^2}.
    \end{align*}
The Jacobian matrix of $\xi_T$ is:
\begin{equation}
    d_{(a,\theta)} \xi_T =
\left(\begin{array}{rrrrrr}
        a_{2} & a_{1} & 0 & 0 & 0 & 0 \\
        a_{3} a_{5} & 0 & a_{1} a_{5} & 0 & a_{1} a_{3} & 0 \\
        a_{4} a_{5} & 0 & 0 & a_{1} a_{5} & a_{1} a_{4} & 0 \\
        0 & a_{3} a_{5} & a_{2} a_{5} & 0 & a_{2} a_{3} & 0 \\
        0 & a_{4} a_{5} & 0 & a_{2} a_{5} & a_{2} a_{4} & 0 \\
        0 & 0 & a_{4} & a_{3} & 0 & 0\\
        0 & 0 & 0 & 0 & 0 & 1
\end{array}\right).
\label{eqn:dxi}
\end{equation}
If we use block notation to express matrices (\ref{eqn:dF}) and (\ref{eqn:dxi}) as $[D~ B]$ and $\left[\begin{array}{cc} E & 0\\ 0 & 1\end{array} \right]$ (where $D$ is $6 \times 6$, $B$ is $6 \times 1$, and $E$ is $6 \times 5$),
then the Jacobian determinant of $F_{k,T} \circ \xi_T: \R^6 \to \R^6$ is
\begin{align*}
    \det(d_{(a,\theta)} (F_{k,T} \circ \xi_T)) &= \det(d_{\xi_T(a,\theta)}F_{k,T} \cdot d_{(a,\theta)} \xi_T)\\
                                        &= \det\left( \left[D ~ B \right]\cdot \left[\begin{array}{cc} E & 0\\ 0 & 1\end{array} \right]\right)\\
                                        &= \det\left( \left[DE ~ B \right]\right)\\
                                        &= 4 D_{12}\bigg(a_{1} a_{2}^{2} a_{3} a_{4}^{2} a_{5}^{2} B_{13}D_{14}D_{23}D_{24} \\*
                                        & \qquad \qquad - a_{1} a_{2}^{2} a_{3}^{2} a_{4} a_{5}^{2} D_{13}B_{14}D_{23}D_{24} \\*
                                        & \qquad \qquad - a_{1}^{2} a_{2} a_{3} a_{4}^{2} a_{5}^{2} D_{13}D_{14}B_{23}D_{24} \\*
                                        & \qquad \qquad + a_{1}^{2} a_{2} a_{3}^{2} a_{4} a_{5}^{2} D_{13}D_{14}D_{23}B_{24} \bigg)D_{34}.
\end{align*}

To prove that $  \det(d_{(a,\theta)} (F_{k,T} \circ \xi_T)) $ is
not the zero function, it suffices to collect terms and show that there is
a nonzero term in this expression.  Here we think of $ \det(d_{(a,\theta)} (F_{k,T} \circ \xi_T)) $ as an element of the ring $\rr(\theta)[a_1,a_2,a_3,a_4,a_5]$.
Ignoring the common factor of $4 D_{12}D_{34} a_1 a_2a_3a_4a_5^2$, we need to show that
\[
a_{2}  a_{4} B_{13}D_{14}D_{23}D_{24} - 
 a_{2} a_{3}  D_{13}B_{14}D_{23}D_{24} -
 a_{1} a_{4} D_{13}D_{14}B_{23}D_{24} + 
 a_{1}a_{3} D_{13}D_{14}D_{23}B_{24}
\]
is not the zero polynomial.  Expanding this as an element of
$\rr(\theta)[a_1,a_2,a_3,a_4,a_5]$, and ignoring the common factor of $\left(2(m-k+1)\frac{1}{4^k}\right)^4$ we see that the lowest order term
in the lexicographic order with $a_1 > a_2 > a_3 > a_4 > a_5$ when $k \geq 2$ is
\[
    -a_1 a_2^2 a_3 a_4^2 a_5^2 \frac{648 k^4(k-1)}{(1 + 8\theta/3)^3(1 + 16 \theta/3)^2}.
\]
The coefficient of this monomial is a nonzero function of $\theta$, so for $k>1$ the Jacobian determinant of 
$F_{k,T} \circ \xi_T$ is nonzero generically,
and thus $d_{(a,\theta)} F_{k,T} \circ \xi_T$ has full column rank.
By Lemma \ref{rank_condition} the map $\phi_{k,T} = F_{k,T} \circ \xi_T$ is 
generically finite-to-one on $\R^6$. Therefore $\phi_{k,T}$ is generically 
finite-to-one on $\Theta_{T} = (0,1]^5 \times \R_{> 0}.$
\end{proof}

The preceding theorem shows that $\phi_{k,T}$ is generically finite-to-one for $n = 4$
leaf trees.
This easily extends to arbitrary trees with $n \geq 4$ leaves.

\begin{corollary}\label{cor:localident}
For $k>1$ and $n\geq 4$, the map $\phi_{k,T}: \R^{2n-2} \to \R^{n(n-1)/2}$ is generically
 finite-to-one. In
particular, it is generically finite-to-one on $\Theta_{T} = (0,1]^{2n-3} \times \R_{> 0}$.
\end{corollary}

\begin{proof}
    We argue by induction. Suppose that $\phi_{k,S}$ is generically finite-to-one for every $n-1$-leaf subtree $S$ of $T$. Let $i$ be a leaf from a cherry of $T$, and let $j$ be a leaf which does not form a cherry with $i$. 
    Then every metric of $T$ in the fiber $\calF_{\phi_{k,T}}(a,\theta) = \phi_{k,T}^{-1}(\phi_{k,T}(a,\theta))$ is determined by its induced metrics on the subtrees $T\backslash i$ and $T \backslash j$. (Here we use $T \backslash i$ to denote the binary phylogenetic tree obtained by deleting leaf vertex $i$, and then suppressing the remaining vertex of the edge incident to $i$.)
    Since $\phi_{k,T\backslash i}$ and $\phi_{k,T \backslash j}$ are finite-to-one generically, $\phi_{k,T}$ is finite-to-one generically.
    The result follows by induction, starting from the $4$-leaf subtrees $Q$ of $T$, for which $\phi_{k,Q}$ is finite-to-one generically by Theorem \ref{local_identifiability}.
\end{proof}

Note that for a $4$ leaf tree with $k = 1$ the Jacobian determinant is just zero, and so
the model is not locally identifiable in this case.  In fact, for monomers, and any tree,
the model $\phi_{1,T}$ is never identifiable.  This can be seen from the special
form of the function $f_{1}$ which is
\[
f_1(x, \theta) = 2m\left(1 - \frac{1}{4} \left(1 + \frac{9x}{3 + 8 \theta} \right)\right).
\]
From this we see that for a given vector $(a_e: e \in E(T),  \theta)$ and a given $\theta_0$, if we can replace
all leaf edge parameters with $b_e = \frac{\sqrt{3+8\theta_0}}{\sqrt{3 + 8 \theta}}a_e$, and make $b_e = a_e$ for internal
edges, then 
\[
\phi_{1,T}(a, \theta) = \phi_{1,T} (b, \theta_0).
\]
In particular, this shows that the dimension of the image of $\phi_{1,T}$ is strictly less than
$2n-2$, so the parameters could not be locally identifiable.

 The following example shows that $\phi_{k,T}$ it is not one-to-one generically on $\Theta_{T}$ for $k=2$ and $n=4$.

\begin{example}
This example shows that for
$k = 2$, $n = 4$ the numerical parameters are not generically identifiability. 
Let $T = 12|34$. Then $\phi_{2,T}(a_1,\theta_1) = \phi_{2,T}(a_2,\theta_2)$ for the following values of $(a_1,\theta_1)$ and $(a_2,\theta_2)$:
\begin{align*}
    (a_1,\theta_1) &= (a_{1,1},a_{1,2},a_{1,3},a_{1,4},a_{1,5},\theta_1) \approx 
    (0.7199,0.6687,0.9623,0.9950,0.9907,0.9146)\\
    (a_2,\theta_2) &= (a_{2,1},a_{2,2},a_{2,3},a_{2,4},a_{2,5},\theta_2) \approx 
     (0.5624,0.5202,0.7642,0.7916,0.9902,0.3400)
\end{align*}

By direct computation, the Jacobian has full rank at $(a_1,\theta_1)$ and $(a_2,\theta_2)$.
By the Implicit Function Theorem, in a union of small open neighborhoods of $(a_1,\theta_1)$ and $(a_2,\theta_2)$,
$\phi_{2,T}$ will be $2$-to-$1$.  This shows that the numerical parameters are not
generically identifiable in this case.
\end{example}


\section{Identifiability of the Tree Topology}\label{sec:treetopo}

In this section, we provide a proof of the generic identifiability of the unrooted tree
topology for Jukes-Cantor $k$-mers for all binary trees with $n\geq 5$ leaves.
The proof is based on applying some ideas from algebraic statistics together with tools from
combinatorial phylogenetics.

To discuss identifiability of the tree parameter, we need to work 
explicitly with the image of $\phi_{k,T}$, that is,   the set of pairwise 
$k$-mer distances that are compatible with our model assumption and the underlying tree
topology $T$.  We denote this set $\calm_{k,T}$.  

\begin{definition}
The tree parameter of the multispecies coalescent model is \emph{generically identifiable}
from $k$-mer distances on trees with $n$ leaves if for all (unrooted) binary trees $T, T'$ on $n$ leaves with
$T \neq T'$,  we have 
\[
\dim( \calm_{k,T} \cap \calm_{k,T'})  <  \min( \dim( \calm_{k,T}), \dim( \calm_{k,T'})).
\]
\end{definition}

An interpretation of the definition is that if we sample a point $p \in \cup_T  \calm_{k,T}$,
then with probability one, there is a unique binary tree $T$ such that $p \in \calm_{k,T}$.  
So generic identifiability of the tree topology means that with probability one, if we have a point
$p$ that came from some tree and some choice of continuous parameters, we can tell which
tree it came from.
One approach to prove the dimension inequality is by using the vanishing
ideals of the sets $\calm_{k,T}$.

\begin{definition}
Let $S \subseteq \rr^d$ and let $\rr[p_1, \ldots, p_d]$ be the polynomial ring in 
indeterminates $p_1, \ldots, p_d$ with real coefficients.  The \emph{vanishing ideal} of $S$
is the set 
\[
I(S)  =  \{ f \in \rr[p_1, \ldots, p_d] :  f(a) = 0  \mbox{ for all  } a \in S  \}.
\]
\end{definition}

For example, 
if $S = \{ (t,t^2) \in \rr^2 :  t \in \rr \}$, then $I(S) =  \langle p_1^2 - p_2 \rangle \subseteq
\rr[p_1, p_2]$.  The polynomial $p_1^2 - p_2$ is called a generator of the ideal $I(S)$.
We refer the reader to \cite{Cox2015} for more background on the necessary algebra,
and \cite{Allman2011} for an example of how this approach is used in studying the identifiability
under the coalescent model in other contexts.

In our case, we will look at the vanishing ideals of the sets $\calm_{k,T} \subseteq \rr^{n(n-1)/2}$
and so the appropriate polynomial ring is $\rr[p_{ij} :  1 \leq i < j \leq n]$.
The following proposition is the key general result about vanishing ideals
that is often used to prove identifiability.

\begin{proposition}\label{prop:identwithideal}
Suppose that $S_1$ and $S_2 \subseteq \rr^d$ are two parametrized sets such that there are nonzero
polynomials $f_1 \in I(S_1) \setminus I(S_2)$ and $f_2 \in I(S_2) \setminus I(S_1)$.
Then 
\[\dim (S_1 \cap S_2)  <  \min(\dim S_1, \dim S_2). \]
\end{proposition}

When the trees have $4$ leaves, the model has $6$ parameters and there are $\binom{4}{2}$
pairwise distances.  Since the numerical parameters are generically locally identifiable in this
case when $k \geq 2$, the vanishing ideals will all be the zero ideals.  This suggests
that the tree parameters might not be generically identifiable when $n = 4$ and $k \geq 2$,
though we have not been able to find specific parameter values realizing this.
On the other hand, when $k =1$ and $n = 4$, the numerical parameters are not identifiable
and the tree parameters \emph{are} identifiable.

\begin{proposition}\label{prop:4leafinvariant}
Let $k = 1$ and $n = 4$, and $T$ and $4$-leaf binary tree.  
Then the vanishing ideal of $I(\calm_{1,T})$ is a principal ideal, and the generator
can be used to distinguish between the different trees.
\end{proposition}

\begin{proof}
Fix the specific tree $T = 12|34$.  The parametrization in the $k = 1$ case has the following form:
\begin{align*}
p_{12} & =   2m\left( 1-\frac{1}{4} ( 1 + a_1 a_2 \frac{9}{3 + 8 \theta})\right) \\
p_{13} & = 2m\left(1- \frac{1}{4} ( 1 + a_1 a_3 a_5 \frac{9}{3 + 8 \theta})\right) \\
p_{14} & = 2m\left(1 - \frac{1}{4} ( 1 + a_1 a_4 a_5 \frac{9}{3 + 8 \theta})\right) \\
p_{23} & = 2m\left(1 - \frac{1}{4} ( 1 + a_2 a_3 a_5 \frac{9}{3 + 8 \theta})\right) \\
p_{24} & = 2m\left(1 - \frac{1}{4} ( 1 + a_2 a_4 a_5 \frac{9}{3 + 8 \theta})\right) \\
p_{34} & =   2m\left( 1-\frac{1}{4} ( 1 + a_3 a_4 \frac{9}{3 + 8 \theta})\right). \\
\end{align*}
These expressions satisfy the single relation:
\begin{equation*}
\left(p_{13} - \frac{3m}{2}\right)\left( p_{24} - \frac{3m}{2}\right) - \left(p_{14} - \frac{3m}{2}\right)\left( p_{23} - \frac{3m}{2}\right). 
\end{equation*}
As the indeterminates that appear in this equation are different from the ones
that occur in the analogous equation for one of the other trees on $4$ leaves, this
shows that we can apply Proposition \ref{prop:identwithideal} to deduce that
the tree parameters are generically identifiable in this case.
\end{proof}

Now we show the analogous result for $n = 5$ leaf trees and for $k \geq 2$.

\begin{proposition}\label{prop:5leafinvariant}
Let $k \geq 2$ and consider the five leaf tree $T$ with nontrivial splits $12|345$ and $123|45$.
Then none of the generators of the vanishing ideal $I( \calm_{k,T})$ involve the indeterminates $p_{12}$ and
$p_{45}$.  Furthermore, there are generators of $I( \calm_{k,T})$ that involve 
the other indeterminates of $\rr[p]$ in a nontrivial way.
\end{proposition}

\begin{proof}
First of all, let us make sense of the statement that
 ``none of the generators of the vanishing ideal 
$I( \calm_{k,T})$ involve the indeterminates $p_{12}$ and
$p_{45}$''.  This is equivalent to saying that  if $(p_{12}, p_{13}, \ldots, p_{35}, p_{45})$ 
is a generic point in $\calm_{k,T}$
then so is $(p_{12} + \epsilon_1, p_{13}, \ldots, p_{35}, p_{45} + \epsilon_2)$ for small $\epsilon_1$
and $\epsilon_2$. 

First, we consider a simplified version of the parametrization where
\[
x_{12} = a_1a_2,  x_{13} = a_1a_3a_6, x_{14} = a_1a_4a_6a_7, x_{15} = a_1a_5a_6a_7,
\]
\[
x_{23} = a_2a_3a_6, x_{24} = a_2a_4a_6a_7, x_{25} = a_2a_5a_6a_7, 
x_{34} = a_3a_4a_7, x_{35} = a_3a_5a_7,x_{45} = a_4a_5.
\]
This is the parametrization of a certain toric ideal associated 
to initial ideals of the Grassmannian \cite{Speyer2004}
and it is known that the vanishing ideal is generated by the following polynomials:
\[
x_{13}x_{24} - x_{14}x_{23}, x_{13}x_{25} - x_{15}x_{23}, 
x_{14}x_{25} - x_{15}x_{24}, x_{15}x_{34} - x_{14}x_{25},
x_{25}x_{34} - x_{24}x_{35}.
\]
It is directly seen that none of these polynomials involve either of the variables
$x_{12}$ or $x_{45}$, in particular, these coordinates can be moved freely while staying in
the variety defined by these equations.  Let $I$ be the ideal defined by these equations and
$\calm_T$ the resulting image of the parametrization for all $a_i$ parameters in $(0,1]$.

Now our model $\calm_{k,T}$ is obtained from  $\calm_T$ by applying the function
\[
f_k(y,\theta) =  2(m - k + 1) \left( 1 - \frac{1}{4^k} \sum_{h = 0}^k  \binom{k}{h} y^h \frac{3^{h+1}}{3 + 8 \theta h} \right)
\]
simultaneously to each coordinate while letting $\theta$ range over $(0,\infty)$.  
Since $f_k(y, \theta)$ is not the zero function and depends nontrivially on $y$,
we can see that if we want to make a small perturbation to the value
$p_{12} = f_k(x_{12}, \theta)$, without changing any of the other $p_{ij}$ values,
 we can do this by perturbing $x_{12}$ and leaving all the other $x_{ij}$'s fixed.
This is possible because no ideal generators of $I$ involved the variables $x_{12}$.  A similar
argument holds with respect to the coordinate $p_{45}$ which proves the first
part of the proposition.

To see that there are equations in $I( \calm_{k,T})$  that do involve all the other variables,
we make two observations.  First, the natural symmetry group of the tree $T$, translates into
symmetries of the ideal $I( \calm_{k,T})$.  In particular, the variables fall into
$3$ orbits under this symmetry group $\{ p_{12}, p_{45} \},$  $\{p_{14}, p_{15}, p_{24}, p_{25} \}, $
and $\{p_{13}, p_{23}, p_{34}, p_{35} \}$. 
The set $\calm_{k,T}$ has dimension $8$, by local identifiability of parameters when $k \geq 2$.
There cannot be any equation involving only the variables  $\{p_{14}, p_{15}, p_{24}, p_{25} \} $,
since such equations would already by implied when looking at $4$ leaf trees (since only
$4$ indices are involved) and we know there are no such relations.  There
also cannot be any relations involving only the variables $\{p_{13}, p_{23}, p_{34}, p_{35} \}$,
because if there were, we could also find such a relation in the ideal $I$ associated
to the parametrization in the $x_{ij}$.  But there is clearly no such
equation.  Since $\calm_{k,T}$ has dimension $8$, there must exist some equation,
any such equation must involve some variables from the set $\{p_{14}, p_{15}, p_{24}, p_{25} \} $
and some from the set  $\{p_{13}, p_{23}, p_{34}, p_{35} \}$.  Then taking the orbit
of such an equation and adding the equations together (with some random coefficients if necessary)
will produce an equation in $I( \calm_{k,T})$ involving all of the variables except
$p_{12}$ and $p_{45}$, as desired.
\end{proof}

\begin{theorem}
For $k=1$, the tree parameter of the $1$-mer multispecies Jukes-Cantor coalescent model
is identifiable for all trees on $n \geq 4$ leaves.  For $k \geq 2$, the tree parameter
of the $k$-mer multispecies Jukes-Cantor coalescent model is identifiable for
all trees on $n \geq 5$ leaves.
\end{theorem}

\begin{proof}
For $k = 1$, Proposition \ref{prop:4leafinvariant} shows that
we can distinguish between $4$ leaf trees using the invariants.  Then we use
the basic fact that if we know the subtrees on all $4$-leaf subsets
we can recover the underlying tree (that is, the quartets in the tree determine the tree, see 
e.g.~\cite{Semple2003}).

For $k \geq 2$, Proposition \ref{prop:5leafinvariant} shows that we
can distinguish between $5$ leaf trees using the invariants of the tree.
Indeed, for each $5$ leaf tree $T$ there will be a distinct set of pair of variables
which are the ones that do not appear in any generator of $I(\calm_{k,T})$.  This 
guarantees that for any $T$, $T'$ that are different $5$ leaf trees that
there are nontrivial polynomials in $I(\calm_{k,T}) \setminus I(\calm_{k,T'})$,
guaranteeing identifiability of the tree parameter from Proposition \ref{prop:identwithideal}.
Once $5$ leaf trees are identified, this tells us all $4$ leaf subtrees in our
underlying tree, which identifies the tree for an arbitrary number of leaves $\geq 5$.
\end{proof}


\section{Identifiability from Combinations of $k$-mer Distances}\label{sec:kand1}
In previous sections, we studied the identifiability of parameters from the expected $k$-mer distances assembled for multiple pairs of taxa. 
Here, we consider a single pair of taxa, and we instead assemble expected $k$-mer distances for two distinct values of $k$, which we denote $k$ and $l$.
In this section, we show that 
numerical parameters are identifiable in this context.
for any set of two or more taxa, when the data consist of
both pairwise $k$-mer distances and pairwise $l$-mer distances for $l \ne k$. 

\begin{theorem}        \label{thm:pairs}
Let $1\le k < l$ and let $\phi: (0,1] \times \R_{>0} \to \R^2$ be the map given by $\phi(x,\theta) = \left( f_k(x,\theta), f_l(x,\theta) \right)$, where $f_k$ is defined by Equation \ref{eqn:rationalparam}. Then $\phi$ is one-to-one.
\end{theorem}

\begin{proof}
    To simplify notation, we reparametrize $f_k$ of Equation \ref{eqn:rationalparam} as follows: Let $y = 3x$, $\xi = 8\theta/3$, and define
\begin{align*}
    g_k(y,\xi) = 1-\frac{1}{4^k}\sum_{h=0}^k { k \choose h } \frac{y^h}{1+\xi h}.
\end{align*}
Then $g_k(y,\xi) = f_k(x,\theta)$. We consider the map $\psi: \R_{>0} \times \R_{>0} \to \R^2$ defined by $\psi(y,\xi) = (g_k(y,\xi), g_l(y,\xi))$.
It suffices to show that $\psi$ is one-to-one. 

By Theorem 7 of \cite{gale1965jacobian}, a function $F: \Omega \to \R^2$ is one-to-one on a rectangular region $\Omega$ if its partial derivatives are continuous and no principal minors of its Jacobian vanish. The partial derivatives of $\psi$ are continuous on the positive orthant, by the following formulas:
\begin{align*}
    \frac{\partial g_k}{\partial \xi} &= \sum_{h=1}^k { k \choose h} \frac{hy^{h}}{(1+\xi h)^2}\\ 
    \frac{\partial g_k}{\partial y} &= -\sum_{h=1}^k { k \choose h} \frac{hy^{h-1}}{1+\xi h}.
\end{align*}

Thus, to show that $\psi$ is one-to-one, it suffices to show that the principal minors of the Jacobian matrix of $\psi$ do not vanish on the positive orthant. This is clear for the $1 \times 1$ minors, so it suffices to show that the Jacobian determinant of $\psi$ does not vanish.

\begin{align*}
    \det (d \psi) (y,\xi) &= 
    \frac{\partial g_k}{\partial y}\frac{\partial g_l}{\partial \xi}
    -\frac{\partial g_l}{\partial y}\frac{\partial g_k}{\partial \xi} = 
    -\frac{\partial g_l}{\partial y}\frac{\partial g_k}{\partial \xi} 
+\frac{\partial g_k}{\partial y}\frac{\partial g_l}{\partial \xi}\\
    &= 
    \sum_{h=1}^k { k \choose h} \frac{hy^{h}}{(1+\xi h)^2}\sum_{i=1}^l { l \choose i} \frac{iy^{i-1}}{1+\xi i}
    -\sum_{h=1}^k { k \choose h} \frac{hy^{h-1}}{1+\xi h}\sum_{i=1}^l { l \choose i} \frac{iy^{i}}{(1+\xi i)^2}\\
    &= \sum_{h=1}^k\sum_{i=1}^l{ k \choose h} { l \choose i} 
    \frac{hy^{h}}{(1+\xi h)^2} \frac{iy^{i-1}}{1+\xi i}
    -\frac{hy^{h-1}}{1+\xi h} \frac{iy^{i}}{(1+\xi i)^2}\\
    &= \sum_{h=1}^k\sum_{i=1}^l{ k \choose h} { l \choose i} \frac{(1+\xi i) - (1+\xi h)}{(1+\xi h)^2(1+\xi i)^2}hiy^{i+h-1}\\
    &= \sum_{h=1}^k\sum_{i=1}^l{ k \choose h} { l \choose i} \frac{\xi hi(i-h)}{(1+\xi h)^2(1+\xi i)^2}y^{i+h-1}.
\end{align*}
Let
\begin{align*}
a_{hi} = { k \choose h} { l \choose i} \frac{\xi hi(i-h)}{(1+\xi h)^2(1+\xi i)^2}.
\end{align*}
Then the coefficient of $y^m$ in $\det(d\psi)$ is:
\begin{equation}
    [y^m] \det (d\psi) = \sum_{i+h-1=m} a_{hi}. \label{eqn:coeffxm}
\end{equation}

To show that $\det (d\psi)$ does not vanish on the positive orthant, it suffices to show that the above sum is positive for all $m$ and for all $\xi > 0$. To see this, we view the $a_{hi}$ as the entries of a matrix $(a_{hi})_{(h,i) \in [k]\times[l]}$. 
For fixed $m$, any term $a_{hi}$ which appears in (\ref{eqn:coeffxm}), lies along the ``cross-diagonal" $h+i-1=m$ of $A$. 
The signs of the entries of this matrix are
\begin{align*}
    (\text{sign}(a_{hi}))_{(h,i) \in [k]\times[l]} =
    \left(\begin{array}{cccccccc}
        0 & + & + & + & + & + & \ldots & + \\
        - & 0 & + & + & + & + & \ldots & + \\
        - & - & 0 & + & + & + & \ldots & + \\
        \vdots & & & \ddots & & & & \vdots \\
        - & - & \hdots & - & 0 & + & \ldots & + \\
    \end{array}\right).
\end{align*}
From the pattern of signs, we see that if $a_{hi}$ is negative, then $h>i$. Since $l>k$, the term $a_{ih}$ is also a term in (\ref{eqn:coeffxm}), and it is positive. Thus it suffices to show that $a_{ih} > -a_{hi}$ for $h>i$. By simple algebra, this is equivalent to the following inequality:
\begin{align*}
    { k \choose i }{ l \choose h }-{ k \choose h }{ l \choose i }>0 \text{ for } h > i.
\end{align*}

This can be verified by expanding the binomial coefficients:
\begin{align*}
    { k \choose i }{ l \choose h }-{ k \choose h }{ l \choose i } &= \frac{k!}{i!(k-i)!}\frac{l!}{h!(l-h)!}-\frac{k!}{h!(k-h)!}\frac{l!}{i!(l-i)!}\\
                                                                  &= \frac{1}{i!h!}\left(\frac{k!}{(k-i)!}\frac{l!}{(l-h)!}-\frac{k!}{(k-h)!}\frac{l!}{(l-i)!}\right)\\
                                                                  &= \frac{1}{i!h!}\left((k)_{i} (l)_{h}-(k)_{h}(l)_{i}\right)\\
                                                                  &= \frac{(l)_i(k)_i}{i!h!}\left((l-i)_{h-i} -(k-i)_{h-i}\right) > 0
\end{align*}

Thus, all principal minors of the Jacobian $d\psi$ are non-vanishing on the positive orthant. Therefore $\psi$ is injective on the positive orthant by Theorem 7 of \cite{gale1965jacobian}. Thus $\phi$ is injective.
\end{proof}

\begin{corollary}
Let $k \neq l$ be positive integers.
Let $T$ be an unrooted tree with $n \geq 2$ leaves and $m$ edges.  The map $\phi_{k,l,T} :  (0,1]^m \times \rr_{> 0}
 \rightarrow  \rr^{n(n-1)}$  with
\[
\phi_{k,l,T}(a,\theta)  =  (\phi_{k,T}(a,\theta), \phi_{l,T}(a, \theta))
\]
is one-to-one.  In particular, the numerical parameters of the model $a, \theta$ are identifiable
given $k$-mer and $l$-mer distances between all pairs of taxa.
\end{corollary}

\begin{proof}
From Theorem \ref{thm:pairs}, we know that we can recover all the products
$x_{ij} = \prod_{e \in P(i,j)} a_e$ for each pair of taxa $i,j$ from the
$k$-mer and $l$-mer vectors.  Since none of the $a_e$ are zero, none of the
$x_{ij}$ are zero either.  Taking logarithms, we have:
\[
-\log x_{ij}  = - \sum_{e \in P(i,j) } \log a_e  = \frac{4}{3}   \sum_{ e \in P(i,j)}  w_e,
\]
so the matrix of $-\log x_{ij}$ is an additive tree distance.  The branch lengths
in an additive tree can be recovered from the pairwise distances.
\end{proof}


\section{Discussion}\label{sec:conclusion}

The expected $k$-mer distance derived here provides the basis for a statistically consistent $k$-mer based method which generalizes the method devised by Allman, Rhodes, and the second author. This generalization extends the $k$-mer method to the case in which sequence blocks correspond to genes, whose gene trees are modeling by the coalescent. 
We have derived our results under a relatively simple setting where there is a global unknown
effective population size $N$ that we use over the entire tree, and we work with the
Jukes-Cantor substitution model.  It would be natural to try to
extend these results to more general settings with a more general substitution model 
and allowing the effective population size to vary over branches of the species tree.

Our identifiability results make a first suggestion for an algorithm for
reconstructing the species tree based on $k$-mer distances from multiple genes.  Namely,
using the results in Section $5$, for distinct positive integers $k$ and $l$, $k$-mer and $l$-mer frequency distributions can be used to
estimate the divergence time between each pair of taxa.  These pairwise distances can
be used in a distance-based method like Neighbor-Joining to reconstruct an
evolutionary tree. 
It remains to be seen how this methodology will perform against other methods for
reconstructing species trees. Our identifiability results are the first step
in showing that methods based on $k$-mers can be derived for these problems.


\section*{Acknowledgments}

Chris Durden was partially supported by the US National Science Foundation (DMS 1615660).
Seth Sullivant was partially supported by the US National Science Foundation (DMS 1615660) and
by the David and Lucille Packard Foundation.

    \bibliographystyle{plain}
    \bibliography{jukes_cantor_kmers}

\end{document}